\documentclass[12pt,a4paper] {article}
\usepackage{latexsym,amsmath,enumerate,amssymb,amsbsy,amsthm}
\usepackage{enumerate,verbatim,tikz}

\newtheorem{theorem}{Theorem}[section]
\newtheorem{corollary}[theorem]{Corollary}
\newtheorem{lemma}[theorem]{Lemma}
\newtheorem{proposition}[theorem]{Proposition}
\theoremstyle{definition}

\newtheorem{remark}[theorem]{Remark}

\numberwithin{equation}{section}

\renewcommand\footnotemark{}

\date{}

\begin{document}

\title{Some Generalizations of Good Integers  and \\ Their Applications in the Study  of \\ Self-Dual Negacyclic Codes} 

\thanks{S. Jitman was supported by the Thailand Research Fund and  Silpakorn University under Research Grant RSA6280042.
    M. Raka was supported by the  Council of Scientific and Industrial Research (CSIR), India, Sanction No. 21(1042)/17/EMR-II.}

\author{Somphong Jitman, Supawadee Prugsapitak  and Madhu Raka
}

\thanks{S. Jitman (Corresponding Author) is with the
    Department of Mathematics, Faculty of Science,
    Silpakorn University,   Nakhon Pathom 73000,  Thailand.
    (email: {sjitman@gmail.com})  }

   \thanks{S. Prugsapitak  is with
  Algebra and Applications Research Unit,
  Department of Mathematics and Statistics,   Faculty of Science, Prince of Songkla University, Hatyai, Songkhla 90110,  Thailand.
    (email: {supawadee.p@psu.ac.th})}         %\\

    \thanks{M. Raka is with the Centre for Advanced Study in Mathematics,
  Panjab University, Chandigarh-160014, INDIA.    (email: mraka@pu.ac.in)}
         %  \\

\maketitle

\begin{abstract}
   Good integers introduced in 1997 form an interesting family of integers that has been continuously studied  due to their rich number theoretical   properties and wide applications. In this paper, we have focused on classes  of  $2^\beta$-good integers, $2^\beta$-oddly-good integers, and $2^\beta$-evenly-good integers which are  generalizations  of good integers.  Properties of such integers have been given as well as their applications in   characterizing and enumerating self-dual negacyclic codes over finite fields.  An alternative proof for  the characterization of  the existence of a  self-dual negacyclic code over finite fields has been given in terms of such generalized good integers.  A general enumeration formula for the number of   self-dual negacyclic codes of length $n$  over  finite fields  has been established.  For some specific lengths, explicit formulas have been provided as well.  Some known results on self-dual negacyclic codes over finite fields can be formalized and viewed as special cases of this work.


\end{abstract}

  \section{Introduction}
 
 For fixed  coprime nonzero integers $a$ and $b$,  a given positive  integer $d$ is  called a  {\em good integer (with respect to $a$ and $b$)}  (see \cite{M1997}) if  there exists a positive integer $k$ such that $d|(a^k+b^k)$. Otherwise, $d$ is called a  {\em bad  integer}. Some properties of the set  $G_{(a,b)}$   of good   integers defined with respect to $a$ and $b$  has been investigated  in \cite{J2017} and \cite{M1997}.   For  a prime power $q$,  the set $G_{(q,1)}$ has been    applied in constructing BCH codes  in \cite{KG1969}.  In \cite{JLX2011} and \cite{JLLX2012},   $G_{(2^l ,1)}$  has been applied  in enumerating  self-dual  cyclic and abelian codes over finite fields.

 In \cite{J2017}, two subclasses of good integers defined with respect to coprime integers $a$ and $b$   have been  introduced.  A positive integer $d$ is said to be {\em oddly-good (with respect to $a$ and $b$)} if $d|(a^k+b^k)$ for some odd integer $k\geq 1$, and {\em evenly-good (with respect to $a$ and $b$)} if $d|(a^k+b^k)$ for some even integer $k\geq 2$. Therefore, $d$ is {good} if it is oddly-good or evenly-good. Denote by $OG_{(a,b)}$ (resp.,    $EG_{(a,b)}$) the set of oddly-good (resp., evenly-good) integers defined with respect to $a$ and $b$.
 In \cite{JLS2013}, some basic properties of  $OG_{(2^l,1)}$ and $EG_{(2^l ,1)}$ have been studied and applied in enumerating Hermitian self-dual abelian  codes over finite fields.   The characterizations  of $OG_{(a,b)}$ and    $EG_{(a,b)}$ was discussed in \cite{J2017} and was applied to the study of the hull of abelian codes. However, there are  errors  in \cite[Proposition 2.1]{J2017} and \cite[Proposition 2.3]{J2017} and  in some subsequent results using these.

 In this paper, we first  give necessary and sufficient conditions for $2^\beta d$ to be a good integer, while correcting  the errors of \cite{J2017}. Then we investigate the  extended classes of $G_{(a,b)}$ defined as follows. For a non-negative integer $\beta$,  a  positive integer $d$ is said to be {\em $2^\beta$-good} (with respect ot $a$ and $b$) if  $2^\beta d\in G_{(a,b)}$. Otherwise, $d$ is said to be {\em $2^\beta$-bad}. In the same fashion,  a positive integer $d$ is said to be {\em $2^\beta$-oddly-good (with respect to $a$ and $b$)} if $2^\beta d\in OG_{(a,b)} $, and {\em $2^\beta$-evenly-good (with respect to $a$ and $b$)} if $2^\beta d\in EG_{(a,b)} $. Therefore, $d$ is   { $2^\beta$-good} if and only if it is $2^\beta$-oddly-good or $2^\beta$-evenly-good.
 For an integer $\beta\geq 0$, denote by  $G_{(a,b)}(\beta)$, $OG_{(a,b)}(\beta)$,  and $EG_{(a,b)}(\beta)$ the sets of $2^\beta$-good, $2^\beta$-oddly-good, and $2^\beta$-evenly-good integers, respectively. 
 Here, we focus on     these $3$ types of generalized  good integers.   Applications of such  generalized good integers in characterizing and enumerating    self-dual negacyclic codes are discussed.
 Brief  history, properties, and applications of (self-dual)  negacyclic codes will be recalled in Section 3. For more details, the readers may refer to  \cite{BR2012}, \cite{BR2013},  \cite{B2008}, \cite{D2012},   \cite{GG2012}, \cite{LI2016}, and references therein.

 The paper is organized as follows.  In Section 2, the family of $2^\beta$-good integers are studied together with their subclasses of    $2^\beta$-evenly-good integers and $2^\beta$-oddly-good integers.  Corrections to  \cite[Propositions 2.1 and 2.3]{J2017}  and its consequence are  also given in this section.    Applications of such  $2^\beta$-good integers in  characterizing  and enumerating   self-dual negacyclic codes are provided in Section 3.
 
 \section{$2^\beta$-Good Integers}
 
 For  pairwise coprime nonzero  integers $a,b$ and $n>0$, let ${\rm ord}_n(a)$ denote the multiplicative order of $a$ modulo $n$. In this case, the multiplicative inverse $b^{-1} $  of  $b$ exists in the multiplicative group $\mathbb{Z}_n^\times$.  Let ${\rm ord}_n(\frac{a}{b})$ denote the multiplicative order of $ab^{-1}$ modulo $n$.   Denote by $2^\gamma ||n$ if $\gamma\geq 0 $ is the largest  integer such that $2^\gamma |n$, i.e.,  $2^\gamma |n$ but $2^{\gamma+1} \nmid n$.

 From the definition of a $2^\beta$-good integer,  a positive integer $d$ is $2^\beta$-good  if and only if  there exists a positive integer $k$ such that $2^\beta d|(a^k+b^k)$. Hence,  for each $\beta\geq 1$, $d$ is $2^{\beta-1}$-good  whenever   it is  $2^\beta$-good. It follows that  $G_{(a,b)}(\beta-1) \supseteq G_{(a,b)}(\beta)$ for all $\beta\geq 1$.

 We note that if $ab$ is even, then  $d\nmid (a^k+b^k) $  for all positive integers $k$ and all  even positive integers $d$.
 Hence, $d\notin G_{(a,b)}$ for all  even integers $d$. Consequently, $G_{(a,b)}(\beta)=\emptyset$ for all $\beta\geq 1$.

 In the following subsections, we assume that $a$ and $b$ are  coprime odd integers. In Subsection 2.1, we rectify the errors of \cite{J2017}. Properties of odd $2^\beta$-good integers are discussed in Subsection 2.2 and arbitrary $2^\beta$-good integers are studied in Subsection 2.3. The subclasses of    $2^\beta$-evenly-good integers and $2^\beta$-oddly-good integers are investigated in Subsection 2.4.

 \subsection{Good Integers: Correction of Results of \cite{J2017}}
 
 %In the following subsections, we assume that $a$ and $b$ are  coprime odd integers. Properties of odd $2^\beta$-good integers are discussed in Subsection 2.1 and arbitrary $2^\beta$-good integers are studied in Subsection 2.2. The subclasses of    $2^\beta$-evenly-good integers and $2^\beta$-oddly-good integers are investigated in Subsection 2.3.
 The errors in \cite{J2017} were caused because of the following false statements
 \begin{equation*} {\rm ord}_{2^{\beta}}(\frac{a}{b}) =2 ~~~ \Rightarrow ~~~ ab^{-1} \equiv -1 ~ {\rm mod}~ 2^{\beta} ~~~~  \text{i.e.,} ~~~2^{\beta}\mid a+b\end{equation*}
 and
 \begin{equation*} {\rm ord}_{d}(\frac{a}{b}) =2k ~~~ \Rightarrow ~~~ (ab^{-1})^k \equiv -1 ~ {\rm mod}~ d,\end{equation*}
 used in  the proofs of \cite[Proposition 2.1]{J2017} and  \cite[Proposition 2.3]{J2017}, respectively,
 where $a$, $b$ and $d\geq 1$  are pairwise  coprime odd integers  and $\beta$ is a positive  integer. It is not difficult to see that  $ {\rm ord}_8(11)=2$  but $11 \not\equiv -1~{\rm mod}~8,$  and $
 {\rm ord}_{15}(11)=2  $ but $11 \not\equiv -1~{\rm mod}~15.$

 First we have a general lemma:
 
 \begin{lemma} \label{ord2k}
     Let $x$ and $d>1$  be   coprime odd integers. If $k$ is the smallest positive integer such that $ x^k \equiv -1 ~ {\rm mod}~ d$,  then ${\rm ord}_{d}(x)=2k.$
 \end{lemma}
 \begin{proof}
     Let $k$ be  the smallest positive integer such that $x^k \equiv -1 ~ {\rm mod}~ d$. Write  $k= 2^{\lambda} k'$, where $\lambda \ge 0$ and $k'$ is an  odd integer.  Since $x^k \equiv -1 ~{\rm mod}~ d$, we have $x^{2k} \equiv 1 ~ {\rm mod}~ d$. Therefore, ${\rm ord}_{d}( x) | 2k.$  Let ${\rm ord}_{d}(x)=  2^{\mu} r $, where  $0\leq \mu \le \lambda+1 $ and $r$ is odd.  Then  $ r| k'$, i.e., $k'=rr'$ for some positive integer $r'$.
     
     Suppose that  $\mu \leq \lambda $. Then $x^{2^{\mu}r} \equiv 1 ~ {\rm mod}~ d$.  It gives that  $x^ {2^\lambda k' } \equiv  x^{2^{\lambda} rr'} \equiv (x^{2^{\mu}r})^{2^{\lambda-\mu}r'} \equiv 1 ~ {\rm mod}~ d,$ but   $x^ {2^\lambda k' } \equiv x^ {k } \equiv -1 ~ {\rm mod}~ d,$  a contradiction, as $d$ is odd. Therefore,  we must have $\mu =\lambda + 1$.

     Write $d = p_1^{e_1}p_2^{e_2}\cdots p_t^{e_t}$,  where $p_1,p_1,\dots,p_t$ are distinct odd  primes and $e_1,e_2,\dots,$ $ e_t$ are some positive integers. Since $1\equiv x^ {{\rm ord}_{d}( x) } =x^{2^\mu r} = x^{2^{\lambda +1}r} ~ {\rm mod}~ p_i^{e_i}$, we have $p_i^{e_i} | (x^{2^{\lambda}r}-1)(x^{2^{\lambda}r}+1)$ for all  $i\in \{1,2,\dots, t\}$.
     Hence, for each  $i\in \{1,2,\dots, t\}$,    either
     $p_i^{e_i} |(x^{2^{\lambda}r'}-1)$ or $p_i^{e_i} | (x^{2^{\lambda}r}+1)$ but not both. If   $p_i^{e_i} | (x^{2^{\lambda}r}-1)$ for some $i\in \{1,2,\dots, t\}                                                                                                                                                        $,  then  $-1\equiv x^k  \equiv  x^{2^{\lambda}k'} \equiv x^{2^{\lambda}rr'}  \equiv  (x^{2^{\lambda}r} )^{r'} \equiv  1 ~ {\rm mod}~ p_i^{e_i}$,  a contradiction.  Hence, $p_i^{e_i} | (x^{2^{\lambda}r}+1)$ for all $i\in \{1,2,\dots, t\}$.  Consequently,  we have $d | (x^{2^{\lambda}r}+1)$, i.e.,                                                                                $x^{\frac{{\rm ord}_{d}( x) }{2}} \equiv x^{ 2^\lambda r} \equiv -1 ~ {\rm mod}~ d$. By the minimality of $k$,  it can be deduced that $k\leq \frac{{\rm ord}_{d}( x) }{2}$.
     Since ${\rm ord}_{d}( x) | 2k$,
     we have  ${{\rm ord}_{d}( x) }=2k$ as desired.
 \end{proof}
 
 \begin{remark}\label{1} The converse of Lemma
     \ref{ord2k} is not always true. The converse holds only if $d$ is an odd prime power or $d=2$. This is so because if  ${\rm ord}_{p^r}(x)=2s$, where $p$ is an odd prime, we have $p^r | (x^s-1)(x^s+1)$. It can not happen that $p^{i} | (x^s-1)$ and $p^j | (x^s+1)$ with $i+j=r, ~i \geq 1, j\ge 1$. Because then $p| (x^s-1)$ and $p | (x^s+1)$ which gives $p | 2$; not possible.  Hence,  either
     $p^r | (x^s-1)$ or $p^r | (x^s+1)$ but not both. If  $p^r | (x^s-1)$, we get ${\rm ord}_{p^r}(x) \geq s$, not possible. Therefore,  $p^r$ must divide $(x^s+1)$.
 \end{remark}
 For each odd integer $x$ and positive integer $\beta$, we note that
 \begin{align}
 \label{eq:ordMod2b}
 {\rm ord }_{2^{\beta}}(x)= \begin{cases} 1 & {\rm if} ~ \beta=1,\\ 2 & {\rm if} ~ \beta \geq 2 ~{\rm and }~ x\equiv -1 ~ {\rm mod}~ 2^{\beta}.
 \end{cases}
 \end{align}

 A correction  of     {\cite[Proposition~2.1 ]{J2017}}  is now given as follows.
 \begin{proposition} \label{prop2.01}Let $a$ and $b$   be  coprime odd  integers  and  let $\beta\geq 1$ be  an  integer. Then  $2^\beta\in G_{(a,b)}$ if and only if  $2^\beta|(a+b)$, i.e., $1\in G_{(a,b)}(\beta)$ if and only if
     $2^\beta|(a+b)$.
 \end{proposition}
 \begin{proof}
     Suppose   $2^\beta \in G_{(a,b)}$.  If  $\beta=1$, then  clearly $2^\beta|(a+b)$   since $a+b$ is even. Let $\beta>1$. Then  $2^\beta|(a^k+b^k)$ for some integer $k\geq 1$
     and so  $4|(a^k+b^k)$. If $k$ is even, then    $a^k\equiv 1 \ {\rm mod\ } 4$ and $b^k\equiv 1 \ {\rm mod\ } 4$  which implies that $(a^k+b^k)\equiv 2 ~{\rm mod }~ 4$, a contradiction.
     It follows that   $k$ is odd.  Since $a^k+b^k=(a+b)\left(\sum\limits_{i=0}^{k-1} (-1)^ia^{k-1-i}b^i \right)$ and  $\sum\limits_{i=0}^{k-1} (-1)^ia^{k-1-i}b^i $ is odd (it being a sum of odd terms taken odd number of times), we have that  $2^\beta |(a+b)$. The converse  is obvious.
 \end{proof}

 The following  results   are  needed in the proof of the correct version of \cite[Proposition 2.3]{J2017}.
 
 \begin{proposition}[{\cite[Theorem 1]{M1997}}] \label{Mor2}  Let $d>1$ be an odd integer. Then $ d \in G_{(a,b)}=G_{(a,b)}(0)$ if and only if there exists an integer $s\geq 1$ such that $2^s|| {\rm ord}_p(\frac{a}{b})$ for every prime $p$ dividing $d$. \end{proposition}
 \begin{proposition}[{\cite[Proposition 2.2]{J2017}}] \label{Jit1}  Let $a,b,d>1$ be pairwise coprime odd integers. Then $ d \in G_{(a,b)}$ if and only if $ 2d \in G_{(a,b)}$. \end{proposition}

 \begin{lemma}[{\cite[Proposition 2]{M1997}}] \label{Mor1}  For an odd prime $p$, ${\rm ord}_{p^{e}}(\frac{a}{b})= {\rm ord}_p(\frac{a}{b})p^{\alpha}$ for some $\alpha\geq 0$. \end{lemma}

 The next proposition is a  correction of   \cite[Proposition 2.3]{J2017}.
 \begin{proposition}  \label{prop2m} Let $a,b$ and $d>1$   be pairwise coprime odd positive integers  and let  $\beta\geq 2$ be  an  integer. Then    $2^\beta d\in G_{(a,b)}$ if and only if   $2^\beta|(a+b)$
     and   $2|| {\rm ord}_{p}(\frac{a}{b})$ for every prime $p$ dividing $d$.
     In this case,    ${\rm ord}_{2^\beta}(\frac{a}{b})=2$  and   $2|| {\rm ord}_{2^\beta d}(\frac{a}{b})$.
 \end{proposition}

 \begin{proof}
     
     Suppose $2^{\beta}d \in G_{(a,b)}$. Let $k$ be the smallest positive integer such that $2^{\beta}d|(a^k + b^k)$.  Then $d|(a^k+b^k) $ and  $2^\beta |(a^k+b^k)$   which implies that  $(ab^{-1})^{2k} \equiv 1 ~{\rm mod} ~d$.  Moreover,   $2^\beta |(a+b)$   and    $k$ must be odd   by    Proposition~\ref{prop2.01} and its proof. Let $k'$ be the smallest  positive integer  such that  $d|(a^{k'}+b^{k'})$.  Then $ {\rm ord}_{d}(\frac{a}{b})=2k'$  by Lemma \ref{ord2k}.  Since  $(ab^{-1})^{2k} \equiv 1 ~{\rm mod} ~d$,  we have  $k'|k$.    Consequently,  $k'$ is odd   and    $(a+b)|(a^{k'}+b^{k'})$.  Hence, $2^\beta d| (a^{k'}+b^{k'})$. By the minimality of $k$, we have  $k= k'$  and  ${\rm ord}_{d}(\frac{a}{b})=2k'=2k$, where $k$ is odd. Let $d = p_1^{e_1}p_2^{e_2}\cdots p_t^{e_t}$ where $p_i$ are odd primes and $e_i \geq 1$. Then, using Lemma ~\ref{Mor1},\vspace{2mm}
     
     $ ~~~~~~~~~~~~~~2k={\rm ord}_d(\frac{a}{b})= {\rm lcm}\big(  {\rm ord}_{p_1}(\frac{a}{b})p_1^{\alpha_1}, {\rm ord}_{p_2}(\frac{a}{b}) p_2^{\alpha_2}, \cdots, {\rm ord}_{p_t}(\frac{a}{b})p_t^{\alpha_t} \big),\vspace{2mm}$

     \noindent where $\alpha_i$ are some non-negative integers. Also $(ab^{-1})^k \equiv -1 ~ {\rm mod}~ p_i$ for all $i,~ 1\leq i \leq t$. Therefore,  ${\rm ord}_{p_i}(\frac{a}{b})$ is even and so $2|| {\rm ord}_{p_i}(\frac{a}{b})$ for each $i$.
     
     Conversely let $2|| {\rm ord}_{p_i}(\frac{a}{b})$ for each $p_i| d$. This gives $2|| {\rm ord}_{p_i^{e_i}}(\frac{a}{b})$ for each $i$ by Lemma \ref{Mor1}. Let $ {\rm ord}_{p_i^{e_i}}(\frac{a}{b})=2r_i$, where $r_i$ is odd. Therefore, by Remark \ref{1},  $(ab^{-1})^{r_i} \equiv -1 ~ {\rm mod}~ p_i^{e_i}$ for all $i, 1\leq i \leq t$. Let $k= {\rm lcm}(r_1,r_2,\cdots,r_t)$, $k$ is odd and let $k=r_ir'_i$. Each of $r'_i$ is also odd. Then $(ab^{-1})^k \equiv (ab^{-1})^{r_i r'_i} \equiv (-1)^{r'_i} \equiv -1 ~ {\rm mod}~ p_i^{e_i}$ for each $i$. Therefore, $(ab^{-1})^k \equiv -1 ~ {\rm mod}~ d$ which implies $ {\rm ord}_{d}(\frac{a}{b})=2k$ by Lemma \ref{ord2k}.  Now $ 2^{\beta}| a+b$ implies $2^{\beta}| a^k+b^k$ as $k$ is odd. Hence $(ab^{-1})^k \equiv -1 ~{\rm mod}~ 2^{\beta}d$, i.e., $2^{\beta}d \in G_{(a,b)}$.

     In this case,      we have $2^\beta|(a+b)$ which implies that
     ${\rm ord}_{2^\beta}(\frac{a}{b})=2$ by \eqref{eq:ordMod2b}. Moreover, $ {\rm ord}_{2^\beta  d}(\frac{a}{b})= {\rm lcm} \left({\rm ord}_{2^\beta  }(\frac{a}{b}), {\rm ord}_{ d}(\frac{a}{b}) \right)=2k $ and $k$ is odd. Hence,
     $2|| {\rm ord}_{2^\beta  d}(\frac{a}{b})$ as desired.
 \end{proof}

 \begin{remark} 
     As a consequence of the above corrections,  the bullets  $(c)$ and $(d)$ of  \cite[Theorem 2.1 and Theorem  3.1]{J2017}  should be  rewritten  as follows.

     \begin{enumerate}
         \item[$(c)$]  {$\beta \ge 2$, $d=1$  and    $ 2^\beta |(a+b)$.}
         
         \item[$(d)$] $\beta \ge  2$, $d\geq 3$,    $ 2^\beta |(a+b)$ and    $2|| {\rm ord}_{p}(\frac{a}{b})$ for every prime $p$ dividing $d$.
     \end{enumerate}
     Note that the above corrections do not affect any other result given in \cite{J2017}.
 \end{remark}

 \subsection{Odd $2^\beta$-Good Integers}

 %In this subsection,  we focus on odd $2^\beta$-good integers. Some useful  results  given  in  \cite{J2017}  and \cite{M1997} are recalled. Subsequently,  additional interesting properties of odd $2^\beta$-good integers are established.

 The characterization  of odd $2^\beta$-good integers greater than $1$   is summarized   in the next proposition.
 \begin{proposition}  \label{oddG}
     Let $a$, $b$  and $d>1$ be  pairwise coprime odd     integers  and let $\beta$ be  a non-negative integer. Then the following statements hold.
     \begin{enumerate}[$1)$]
         \item  ({\cite[Theorem 1]{M1997}})  If $\beta=0$, then $d\in G_{(a,b)}=G_{(a,b)}(\beta)$   if and only if there exists $s\geq 1$ such that $2^s|| {\rm ord}_{p}(\frac{a}{b})$ for every prime $p$ dividing~$d$.
         \item (\cite[Proposition 2.2]{J2017})  If $\beta=1$, then    $d\in G_{(a,b)}(\beta)$
         if and only if        $d\in G_{(a,b)}$.
         \item (Proposition~\ref{prop2m})  If  $\beta\geq 2$, then    $d\in G_{(a,b)}(\beta)$ if and only if
         $2^\beta|(a+b)$  and $2|| {\rm ord}_{p}(\frac{a}{b})$ for every prime $p$ dividing $d$.
     \end{enumerate}
 \end{proposition}

 For  coprime odd  integers $a$ and $b$, it is obvious that $1$ is an element in both $ G_{(a,b)}(0)$ and $G_{(a,b)}(1)$.   For $\beta\geq 2$, we have the following characterization.
 
 \begin{corollary} \label{corBad}
     Let $a$ and $b$   be  coprime odd  integers  and  let $\beta\geq 2$ be  an  integer. Then the following statements are equivalent.
     \begin{enumerate}[$1)$]
         \item $1\notin G_{(a,b)}(\beta)$.
         \item $2^\beta \nmid (a+b)$.
         
         \item $n \notin G_{(a,b)}(\beta)$ for all odd natural numbers $n$.
     \end{enumerate}
 \end{corollary}
 \begin{proof}  From Proposition~\ref{prop2.01},  we have that 1) and 2) are  equivalent.
     
     To prove 2) implies 3), assume that $2^\beta \nmid (a+b)$.  By Proposition~\ref{oddG}, we therefore have  $n \notin G_{(a,b)}(\beta)$ for all odd natural numbers $n$. That 3)  implies 1) is clear.
 \end{proof}
 
 Next, we consider  a  product of two odd  $\beta$-good integers and the  divisors of an odd $\beta$-good integer.

 \begin{proposition}\label{factorG}
     Let $a$ and  $b$    be    coprime odd     integers. Let $\beta$ be  a non-negative  integer and let $c$ and  $d$ be odd positive integers.
     \begin{enumerate}[$1)$]
         \item  If      $cd \in G_{(a,b)}(\beta)$, then   $c\in G_{(a,b)}(\beta)$ and $d \in G_{(a,b)}(\beta)$.
         \item If $\beta\geq 2$, then     $cd \in G_{(a,b)}(\beta)$ if and only if   $c\in G_{(a,b)}(\beta)$ and $d \in G_{(a,b)}(\beta)$.
     \end{enumerate}
 \end{proposition}
 \begin{proof} 1) and the sufficient part of 2) follow directly from the definition of $G_{(a,b)}(\beta)$.
     The necessary part of 2) follows from Proposition~\ref{oddG}.
 \end{proof}

 \begin{remark}
     We note that the converse of 1) in Proposition  \ref{factorG} does not need to be  true for $ \beta\in \{ 0, 1\}$.   Using Proposition~\ref{oddG} and a direct calculation,  we have the following examples.
     \begin{enumerate}
         \item $ 5\in  G_{(3,1)}(0)$ and $ 7\in  G_{(3,1)}(0)$ but $ 5\times 7\notin  G_{(3,1)}(0)$.
         \item  $ 5\in  G_{(3,1)}(1)$ and $ 7\in  G_{(3,1)}(1)$ but $ 5\times 7\notin  G_{(3,1)}(1)$.
     \end{enumerate}
 \end{remark}

 The following corollary can be obtained immediately from Proposition~\ref{factorG}.
 \begin{corollary}
     Let $a$, $b$  and $d>1$ be  pairwise coprime odd     integers  and let $\beta$ be  a non-negative   integer.   Then the following statements hold.
     \begin{enumerate}
         \item     If      $d \in G_{(a,b)}(\beta)$,  then   $c\in G_{(a,b)}(\beta)$ for all proper/prime  divisors $c$ of $d$.
         \item     If   $\beta\geq 2$, then  $ d\in G_{(a,b)}(\beta)$ if and only if  $c\in G_{(a,b)}(\beta)$ for all proper/prime divisors $c$ of $d$.
     \end{enumerate}
 \end{corollary}

 \subsection{Arbitrary $2^\beta$-Good Integers}
 
 Here, we focus on arbitrary $2^\beta$-good integers and derive the following results.

 \begin{lemma} \label{Gb+1}
     Let $a$ and $b$   be    coprime odd     integers  and let $\beta$   be   a positive integer such that $2^\beta||(a+b)$.  Then the following statements hold.
     \begin{enumerate}[$1)$]
         \item   $ G_{(a,b)}(\beta)\ne \emptyset$ and every element in $ G_{(a,b)}(\beta)$ is odd.
         \item  $ G_{(a,b)}(\beta+1)=\emptyset$.
     \end{enumerate}
 \end{lemma}
 \begin{proof} Clearly,  $ 1\in G_{(a,b)}(\beta)$.  Next, suppose that $ G_{(a,b)}(\beta)$ contains an even integer, denoted it by $2^id$ for some positive integer $i$ and odd positive  integer $d$.  Then $d\in   G_{(a,b)}(\beta+i)$. Since $\beta+i\geq 2$,  we have $2^{\beta+i}|(a+b)$ by 3) of Proposition~\ref{oddG}. This is  a contradiction. Hence, 1) is proved.

     To prove 2), suppose that  $ G_{(a,b)}(\beta+1)\ne \emptyset$.   Let  $ d \in G_{(a,b)}(\beta+1)$ so $2^{\beta +1}d \in G_{(a,b)}$.  Since  $\beta+1\ge 2$,    we have
     $ {2^{\beta+1}}|(a+b)$  by  3) of Proposition~\ref{oddG}. This is a contradiction.
 \end{proof}

 The characterization of the set $ G_{(a,b)}(\beta )$ is  given in the next theorem.
 
 \begin{theorem} \label{thmSet} Let $a$ and $b$ be  coprime  odd integers and let $\beta\geq 0$ be an integer. Let $\gamma$ be a positive  integer such that $2^\gamma||(a+b)$. Then the following statements hold.
     \begin{enumerate}[$1)$]
         
         \item   If $\gamma<\beta$, then  $ G_{(a,b)}(\beta )=\emptyset$.
         
         \item If $2\leq \beta\leq \gamma$, then
         \begin{align*}G_{(a,b)}(\gamma ) =\{d \in \mathbb{N} \mid d=1 \text{  or  } &d \text{  is odd such that } \\ & 2|| {\rm ord}_{p}(\frac{a}{b})~ \text{for every  prime}~        p ~\text{dividing}~d\} \end{align*}
         and for  $2\leq \beta<\gamma$
         \begin{align*}
         G_{(a,b)}(\beta )   &= \bigcup _{i=0}^{\gamma-\beta} \left\{d 2^{i}\mid d \in  G_{(a,b)}(\gamma ) \right\}  \\
         &=G_{(a,b)}(\beta+1 )\cup  \left\{d 2^{\gamma-\beta}\mid d \in  G_{(a,b)}(\gamma ) \right\}
         %\\&=G_{(a,b)}(\beta+1 )\cup  \left\{d 2^{\gamma-\beta}\mid d \in  G_{(a,b)}(\beta+1 ) \right\}
         \end{align*}

         \item   If $ \beta \in \{0, 1\}$, then
         \begin{align*}G_{(a,b)}(1 )   =  \{2d\mid d\in G_{(a,b)}(2) \} \cup \{ d \in \mathbb{N} \mid & d=1  \text{ or $d$ is odd and  there exists}\\
         & s\geq 1~ \text{such that}~ 2^s|| {\rm ord}_{p}(\frac{a}{b}) \\&  \text{for  every  prime }       p ~\text{dividing}~d \}\end{align*}
         and
         $$G_{(a,b)}(0)   = G_{(a,b)}(1 ) \cup \left\{  2d \mid d\in G_{(a,b)}(1 )  \right\}.$$
     \end{enumerate}
 \end{theorem}
 \begin{proof}
     From Lemma \ref{Gb+1}, 1) is clear.

     To prove 2), assume that $2\leq \beta \leq \gamma$.     Since $2^{\gamma} || (a+b)$,  {we have  $2^\beta|(a+b)$.}
     By Lemma \ref{Gb+1},  $G_{(a,b)}(\gamma)$  is the set of odd $2^\gamma$-good integers and the result follows from Proposition~\ref{prop2m}. Let now $ \beta < \gamma$.  Let $c\in G_{(a,b)}(\beta )$. Let $c=2^id$ for some integer $i\geq 0$ and some odd integer $d$. It follows that  $ 2^{\beta+i} d =2^\beta c \in G_{(a,b)}(0 )$. Hence,  $  2^{\beta+i} \in G_{(a,b)}(0 )$ and $0\leq i \leq \gamma-\beta$  because   $2^{\gamma} || (a+b)$.  Since $2\leq \beta$, we have $ d \in G_{(a,b)}(2 )$ which implies that $ d \in G_{(a,b)}(\gamma )$ by Proposition~\ref{prop2m}.  Therefore,  $c=2^id\in \left\{d 2^{i}\mid d \in  G_{(a,b)}(\gamma ) \right\}$ for some $0\leq i \leq \gamma-\beta$.  The reverse inclusion is clear.
     Therefore,  $G_{(a,b)}(\beta )   = \bigcup _{i=0}^{\gamma-\beta} \left\{d 2^{i}\mid d \in  G_{(a,b)}(\gamma ) \right\}.$
     
     Now \begin{align*} G_{(a,b)}(\beta ) &  = \bigcup _{i=0}^{\gamma-\beta-1} \left\{d 2^{i}\mid d \in  G_{(a,b)}(\gamma ) \right\} \cup \left\{d 2^{\gamma-\beta}\mid d \in  G_{(a,b)}(\gamma ) \right\}\\&= G_{(a,b)}(\beta+1 )\cup  \left\{d 2^{\gamma-\beta}\mid d \in  G_{(a,b)}(\gamma ) \right\}.
     \end{align*}
     
     Next, we prove 3). If $\gamma=1$,  then $G_{(a,b)}(2 )=\emptyset $  by Lemma \ref{Gb+1}.  The results follow from 1)-2) of Proposition~\ref{oddG}.  Assume that $\gamma\geq 2$.
     
     \noindent     {\bf Case I:} $\beta=1$.   Let $c\in G_{(a,b)}(1)$.   If $c$ is odd,  then $c\in G_{(a,b)}( 0)=\{ d \in \mathbb{N} \mid  d=1  \text{ or $d$ is odd and  there exists } s\geq 1~ \text{such that}~ 2^s|| {\rm ord}_{p}(\frac{a}{b})   \text{ for  every  prime }       p $ $\text{dividing}~d \}$  by  1)-2) of Proposition~\ref{oddG}.  Suppose that $c=2d$ is even for some positive integer $d$. Since $2d=c\in G_{(a,b)}(1)$, we have $d\in G_{(a,b)}(2)$ and hence $c=2d\in \{2d\mid d\in G_{(a,b)}(2) \}$.
     The reverse inclusion is clear.
     
     \noindent     {\bf Case II:} $\beta=0$.  Let  $c\in G_{(a,b)}(0) $. If $c$ is odd,  then $c\in G_{(a,b)}( 1) $ by  2) of Proposition~\ref{oddG}.   Suppose that $c=2d$ is even for some positive integer $d$.  If  $2d=c\in G_{(a,b)}(0)$, we have $d\in G_{(a,b)}(1)$ and hence $c=2d\in \{2d\mid d\in G_{(a,b)}(1) \}$.  The reverse inclusion is clear.
 \end{proof}

 \begin{corollary} \label{coroll2.4} Let $a$ and $b$ be coprime odd  integers and let $\gamma$ be a positive integer such that $2^\gamma||(a+b)$.  Then
     \begin{align*} G_{(a,b)}=G_{(a,b)}(0) \supsetneq &G_{(a,b)}(1)  \supsetneq G_{(a,b)}(2)     \supsetneq  \cdots  \supsetneq    G_{(a,b)}(\gamma)   \supsetneq    G_{(a,b)}(\gamma+1)  =\emptyset.
     \end{align*}
 \end{corollary}
 \begin{proof}
     From the definition, it is not difficult to see that  $G_{(a,b)}(i) \supseteq G_{(a,b)}(i+1) $ for all $0\leq i \leq \gamma$. From Theorem~\ref{thmSet},   we have $G_{(a,b)}(\gamma+1) =\emptyset$.  Again, by   Theorem~\ref{thmSet},  it can be seen that  $2^{\gamma-i} \in   G_{(a,b)}(i) \setminus G_{(a,b)}(i+1) $  for all $0\leq  i \leq \gamma-1$ and  $1\in  G_{(a,b)}(\gamma) \setminus G_{(a,b)}(\gamma+1) $.
 \end{proof}

 \subsection{$2^\beta$-Oddly-Good and $2^\beta$-Evenly-Good  Integers}

 In this subsection,  we focus on families of  $2^\beta$-oddly-good  and $2^\beta$-evenly-good integers.
 
 First, recall that a positive integer $d$ is said to be {\em $2^\beta$-oddly-good (with respect to $a$ and $b$)} if $2^\beta d\in OG_{(a,b)} $, and {\em $2^\beta$-evenly-good (with respect to $a$ and $b$)} if $2^\beta d\in EG_{(a,b)} $.
 Useful  characterization of   $OG_{(a,b)}(0)=OG_{(a,b)}$ and of $EG_{(a,b)}(0)=EG_{(a,b)}$ from \cite{J2017} are recalled as follows (These results are not affected by the  errors of \cite{J2017}  discussed in Section 2.1).

 \begin{proposition}[{\cite[Proposition 3.2]{J2017}}]
     \label{odd-good} Let $a$ and $b$ be coprime non-zero integers and let $d>1$ be an  odd  integer. Then   $d\in OG_{(a,b)}(0)=OG_{(a,b)}$   if and only if $2|| {\rm ord}_{p}(\frac{a}{b})$ for every prime $p$ dividing $d$.
 \end{proposition}

 \begin{corollary}[{\cite[Corollary 3.2]{J2017}}]
     \label{odd-good2} Let $a$ and $b$ be coprime non-zero integers and let $d$ be an  odd positive  integer.  Then the following statements hold.
     \begin{enumerate}[$1)$]
         \item    $d\in OG_{(a,b)}(0)= OG_{(a,b)}$ if and only if  $2d\in OG_{(a,b)}(0)$ if and only if  $d\in OG_{(a,b)}(1)$.

         \item  For each $\beta \ge 2$,  $  d\in OG_{(a,b)}(\beta)$ if and only if $ d\in G_{(a,b)}(\beta)$.
     \end{enumerate}
 \end{corollary}

 \begin{proposition}[{\cite[Proposition 3.3]{J2017}}]
     \label{odd-even} Let $a$ and $b$ be coprime nonzero  integers and let $d>1$ be an  odd  integer. Then
     $d\in EG_{(a,b)}(0)= EG_{(a,b)}$    if and only if there exists $s\geq 2$ such that $2^s|| {\rm ord}_{p}(\frac{a}{b})$ for every prime $p$ dividing $d$.
 \end{proposition}

 From the definitions, we have $G_{(a,b)}(\beta) =OG_{(a,b)}(\beta) \cup EG_{(a,b)}(\beta) $ for all non-negative integers $\beta$.  In many cases, the following theorem shows that $EG_{(a,b)}(\beta) =\emptyset$.

 \begin{theorem}\label{OEG} Let $a$ and $b$ be coprime odd  integers and let $\beta$ be a non-negative integer.  Then $EG_{(a,b)}(\beta) \ne \emptyset $  if and only if $\beta\in\{0,1\}$, i.e. $EG_{(a,b)}(\beta)= \emptyset $  if and only if $\beta \geq 2$.
 \end{theorem}
 \begin{proof}
     Assume that $EG_{(a,b)}(\beta) \ne \emptyset $.  Let $c\in EG_{(a,b)}(\beta) $. Then $2^\beta c\in  EG_{(a,b)}$ which implies that $2^\beta \in  EG_{(a,b)}$.  It follows that $2^\beta \in  G_{(a,b)}$  and $2^\beta|(a+b)$ by Proposition~\ref{prop2.01}. Then $2^\beta \in  OG_{(a,b)}$.  By  \cite[Proposition 3.1]{J2017},  any positive integer greater than $2$ can be either oddly-good or evenly-good but not both. Hence, $\beta\in \{0,1\}$ as desired.
     
     For the converse, it is clear that $1\in  EG_{(a,b)}(0) $ and $1\in EG_{(a,b)}(1) $. \end{proof}
 
 From Theorem~\ref{OEG}, we have that $OG_{(a,b)}(\beta) =G_{(a,b)}(\beta) $  for all $\beta\geq 2$ and they are determined in  Theorem~\ref{thmSet}.  Next we  investigate    $OG_{(a,b)}(\beta)$ and  $EG_{(a,b)}(\beta)$ for $\beta \in \{0, 1\}$.
 \begin{theorem}
     Let $a$ and $b$ be coprime  odd  integers an let $\gamma$ be a positive  integer such that $2^\gamma||(a+b)$.  Then the following statements hold.
     \begin{enumerate}[$1)$]
         \item $OG_{(a,b)}(1)= \bigcup \limits _{i=0}^{\gamma-1} \left\{d 2^{i}\mid d \in O G_{(a,b)}(\gamma ) \right\}$ for $\gamma \geq 2$, where \\$ OG_{(a,b)}(\gamma ) =\{d \in \mathbb{N} \mid d=1 \text{  or   }  d \text{  is odd such that }  2|| {\rm ord}_{p}(\frac{a}{b})~ \text{for every  prime}$ $       p ~\text{dividing}~d\} $
         \item  $EG_{(a,b)}(1) = \{ d \in \mathbb{N} \mid  d=1  \text{ or $d$ is odd and  there exists } s\geq 2~ \text{such that }$ $ 2^s|| {\rm ord}_{p}(\frac{a}{b})   \text{ for } \\ ~~~~~~~~~~~~~~~~~~\text{ every  prime }       p $ $\text{dividing}~d \}$.
         \item  $OG_{(a,b)}(0)= OG_{(a,b)}(1) \cup  \{2d\mid d\in OG_{(a,b)}(1)\}$.
         \item  $EG_{(a,b)}(0)=EG_{(a,b)}(1)\cup \{2d\mid d\in EG_{(a,b)}(1)\}$.
     \end{enumerate}
 \end{theorem}
 \begin{proof} If $\gamma =1 $, every element in  $ O G_{(a,b)}( 1)$ is odd  by Lemma \ref{Gb+1} as $ O G_{(a,b)}( 1)\subseteq G_{(a,b)}( 1)$.  Hence, the characterization of  $O G_{(a,b)}( 1)$ follows from  Corollary \ref{odd-good2} and Proposition \ref{odd-good}.  Next, assume that   $\gamma\geq 2$. Then we  have   $O G_{(a,b)}(\gamma ) =G_{(a,b)}(\gamma ) $ by Theorem \ref{thmSet}.
     Let $c\in OG_{(a,b)}(1)$. Let $c=2^id$ for some integer $i\geq 0$ and some odd integer $d$. It follows that  $ 2^{1+i} d =2 c \in OG_{(a,b)}(0 )$. Hence,  $  2^{1+i} \in OG_{(a,b)}(0 )$ and $0\leq i \leq \gamma-1$  because   $2^{\gamma} || (a+b)$.
     Subsequently, we have $ d \in OG_{(a,b)}(0 )$.   Since $d $ is odd, it can be concluded that  $ d \in \{d \in \mathbb{N} \mid d=1 \text{  or   }  d \text{  is odd such that }  2|| {\rm ord}_{p}(\frac{a}{b})~ \text{for every} $ $\text{prime}   $ $     p ~\text{dividing}~d\} = OG_{(a,b)}(\gamma )$ by Proposition~\ref{odd-good}.  Therefore,  $c=2^id\in \left\{d 2^{i}\mid d \in  OG_{(a,b)}(\gamma ) \right\}$ for some $0\leq i \leq \gamma-1$.  The reverse inclusion is clear. Hence, the proof of 1) is completed.

     Every element  in $EG_{(a,b)}(1) $ is odd. For if $x=2y \in EG_{(a,b)}(1)$, then  $2^2y \in EG_{(a,b)}(0)$ and so  $y \in EG_{(a,b)}(2)$. But from Theorem~\ref{OEG}, $EG_{(a,b)}(2)=\emptyset$.  Now 2) follows immediately from Proposition \ref{odd-even}.
     
     Using arguments similar to those in the proof of 3) in Theorem \ref{thmSet}, 3) and 4) follow.
 \end{proof}

 From Corollary~\ref{coroll2.4} and Theorem~\ref{OEG}, the results can be summarized in the following diagram.
 
 \setlength{\tabcolsep}{1pt}
 \renewcommand{\arraystretch}{1.5}
 \begin{center}
     \begin{tabular}{cccccccccccccc}$G_{(a,b)}$&$=$&$G_{(a,b)}(0)$ &&$\supsetneq$ &$G_{(a,b)}(1) $& $\supsetneq $&$G_{(a,b)}(2)  $    &$\supsetneq $ &$\cdots $ &$\supsetneq  $  &$G_{(a,b)}(\gamma) $  &$\supsetneq $  & $G_{(a,b)}(\gamma+1)  =\emptyset$\\
         &&$||$&&&$|| $&&$||$&&&&$||$&&$||$\\
         & & \rotatebox[origin=c]{90}{$OG_{(a,b)}(0) \cup EG_{(a,b)}(0)$~} & && \rotatebox[origin=c]{90}{$OG_{(a,b)}(1) \cup EG_{(a,b)}(1)$~} && \rotatebox[origin=c]{90}{~~~~~~~~~~~~~~~$OG_{(a,b)}(2) $}  &&$\cdots$& & \rotatebox[origin=c]{90}{~~~~~~~~~~~~~~~~$OG_{(a,b)}(\gamma)$ }&&  \rotatebox[origin=c]{90}{~~~~~~~~~~$OG_{(a,b)}(\gamma+1) $}
     \end{tabular}
 \end{center}

 \section{Self-Dual Negacyclic Codes}
 
 In this section, we focus on applications of $2^\beta$-good  and $2^\beta$-oddly-good integers in the characterization and enumeration of self-dual negacyclic codes.   In Subsection 3.1,   brief history and basic properties of  negacyclic codes  are recalled. It is  followed by the characterization and enumeration  of Euclidean  (resp., Hermitian) self-dual negacyclic codes in Subsection 3.2  (resp., Subsection 3.3).

 \subsection{Negacyclic Codes}
 For a prime  $p$ and a positive integer $l$, denote by $\mathbb{F}_{p^l}$ the finite field of $p^l$ elements. A {\em linear code} $C$ of length $n$ over $\mathbb{F}_{p^l}$ is defined to be a subspace of the $\mathbb{F}_{p^l}$-vector space $\mathbb{F}_{p^l}^n$.  The {\em Euclidean dual} of a linear code $C$  is defined to be \[C^{\perp_{\rm E}}=\{\boldsymbol{v}\in \mathbb{F}_{p^l}^n\mid \langle \boldsymbol{v} , \boldsymbol{c} \rangle_{\rm E}=0\text{ for all } \boldsymbol{c}\in C\},\]
 where $\langle \boldsymbol{v} , \boldsymbol{u} \rangle_{\rm E}:=\sum_{i=1}^n v_iu_i$ is the Euclidean inner product between $\boldsymbol{v} =(v_1,v_2,\dots, v_n) $ and  $\boldsymbol{u} =(u_1,u_2,\dots, u_n)$ in $\mathbb{F}_{p^l}^n$.
 Over $\mathbb{F}_{p^{2l}}$, the {\em Hermitian  dual} of  a linear code $C$    can be defined as well and it  is defined to be \[C^{\perp_{\rm H}}=\{\boldsymbol{v}\in \mathbb{F}_{{p^{2l}}}^n\mid \langle \boldsymbol{v} , \boldsymbol{c} \rangle_{\rm H}=0\text{ for all } \boldsymbol{c}\in C\},\]
 where $\langle \boldsymbol{v} , \boldsymbol{u} \rangle_{\rm H}:=\sum_{i=1}^n v_iu_i^{p^{l}}$ is the  Hermitian inner product between $\boldsymbol{v}   $ and  $\boldsymbol{u}  $ in $\mathbb{F}_{{p^{2l}}}^n$.  A code $C$ is said to be {\em Euclidean self-dual}
 (resp.  {\em Hermitian self-dual}) if  $C=C^{\perp _{\rm E}}$ (resp., $C=C^{\perp _{\rm H}}$).
 
 A linear code $C$ is said to be {\em  negacyclic} if it is invariant under the right negacyclic shift. Precisely,  a linear code $C$ is negacyclic if and only if   \[(-c_{n-1},c_0,c_1, \dots, c_{n-2})\in C \text{ whenever } (c_0,c_1,c_2, \dots, c_{n-1})\in C.\]  Due to their rich algebraic structures and wide applications, negacyclic codes with self-duality have been of interest and extensively studied (see \cite{BR2012}, \cite{BR2013},  \cite{B2008}, \cite{D2012},   \cite{GG2012}, and \cite{LI2016}).
 Unlike the cyclic case in \cite{JLX2011}, self-dual negacyclic code of length $n$ over $\mathbb{F}_{p^l}$ exists only if $p$  is odd.
 The characterization  for the existence of a Euclidean  self-dual negacyclic code  over  $\mathbb{F}_{p^l}$ has been given in \cite{B2008}.
 In    \cite{D2012},  algebraic structure of repeated root Euclidean self-dual  negacyclic codes of length $n=2p^r$ over $\mathbb{F}_{p^l}$ has been  studied.  In \cite{BR2012} and  \cite{BR2013}, all simple root self-dual negacyclic codes of lengths $2^\nu$  and    $2q^t$  over $\mathbb{F}_{p^l}$ have been determined in terms of their generator polynomials, where $q$ is an odd prime different from  $p$. In the said papers, the   enumeration of such self-dual negacyclic codes   has been given as well.
 All Euclidean self-orthogonal negacyclic codes of length $q^\nu$ and  $2q^\nu$  over $\mathbb{F}_{p^l}$ have been determined in \cite{BR2013} and \cite{R2015}.
 Euclidean self-dual negacyclic codes of oddly-even length  have  been studied  in \cite{GG2012}.
 In \cite{LI2016},  construction and enumeration for Euclidean self-dual negacyclic code of length $2^\nu p^r$  have been provided together with a general concept for enumeration. However, there are no explicit formulas.

 It is well known that every negacyclic code $C$ of length $n$ over $\mathbb{F}_{p^l}$ can be viewed as an (isomorphic) ideal in the principal ideal ring $\mathbb{F}_{p^l}[x]/\langle x^n+1\rangle$ uniquely generated by a monic divisor $g(x)$  of $x^n+1$. Such  polynomial is called the {\em generator polynomial} of   $C$. For a monic polynomial $f(x)=\sum_{i=0}^kf_ix^i$  of degree $k$ in $\mathbb{F}_{p^l}[x]$ with $f_0\ne 0$,  the {\em reciprocal polynomial} of $f(x)$ is defined to be $f^*(x):=f_0^{-1} x^k \sum_{i=0}^kf_i(1/x)^i$.  In $\mathbb{F}_{p^{2l}}[x]$,  the {\em conjugate reciprocal polynomial} of $f(x)$ is defined to be $f^\dagger (x):=f_0^{-p^l} x^k \sum_{i=0}^kf_i^{p^l}(1/x)^i$.  A polynomial $f(x)$ is  called {\em self-reciprocal} (resp. {\em self-conjugate-reciprocal}) if $f(x)=f^*(x)$ (resp., $f(x)=f^\dagger(x)$). Otherwise, $f(x)$ and  $f^*(x)$ form a {\em reciprocal polynomial pair} (resp, $f(x)$ and  $f^\dagger(x)$ form a {\em conjugate-reciprocal polynomial pair}).  In \cite[Proposition 2.4]{D2010} and \cite[Proposition 2.3]{YC2015}, it has been shown that the Euclidean and Hermitian duals of   a negacyclic code $C$  over finite fields are again negacyclic codes.  Moreover, if  $C$ is  a negacyclic code with the generator polynomial $g(x)$, then it is Euclidean self-dual (resp., Hermitian self-dual) if and only if $g(x)=h^*(x)$ (resp., $g(x)=h^\dagger(x)$), where $h(x)=\frac{x^n+1}{g(x)}$.
 
 Consider an odd prime $p$ and $n=2^\nu p^rn^\prime$, where $\nu\geq 0$ and $r\geq 0$ are integers and $ n^\prime$ is an odd positive integer such that $p\nmid n^\prime$.  We have
 \begin{align}\label{eq-firstfact}
 x^n+1=(x^{2^\nu n^\prime}+1)^{p^r}&=\left(\dfrac{x^{2^{\nu+1} n^\prime}-1}{x^{2^\nu n^\prime}-1}\right)^{p^r}    =\left(\dfrac{   \prod \limits_{d| 2^{\nu+1}n^\prime} Q_d(x) }{  \prod\limits _{d| 2^{\nu}n^\prime} Q_d(x)}\right)^{p^r}   \notag \\
 &=\left( \prod _{d| n^\prime} Q_{d2^{\nu+1}}(x)  \right)^{p^r} ,
 \end{align}
 where $Q_{d2^{\nu+1}}(x) := \prod\limits_{\substack{{1\leq i\leq d2^{\nu+1} }\\ {\gcd(i,d 2^{\nu+1})=1}}} (x-\omega^i)$ is the $d2^{\nu+1}$th cyclotomic polynomial and $\omega$ is a primitive  $d2^{\nu+1}$th root of unity.

 \subsection{Euclidean Self-Dual Negacyclic Codes}
 
 In this subsection, we focus on the characterization and enumeration of Euclidean self-dual negacyclic codes over $\mathbb{F}_{p^l}$. General properties of Euclidean self-dual negacyclic codes of any lengths are given in 3.2.1.  Euclidean self-dual negacyclic codes  of some specific lengths are discussed in 3.2.2.
 
 \subsubsection{Euclidean Self-Dual Negacyclic Codes}

 From \cite[Theorem 2.47]{LNBook}, it has been shown that the $d2^{\nu+1}$th cyclotomic polynomial  $Q_{d2^{\nu+1}}(x)$ in \eqref{eq-firstfact} can be factorized into  a product of $\frac{\phi(d2^{\nu+1})}{{\rm  ord}_{d2^{\nu+1}}(p^l)}$ distinct monic irreducible polynomials of the same degree in $\mathbb{F}_{p^l}[x]$, where  $\phi$  is the Euler's totient function.
 
 The results in \cite[Lemma 1]{JLX2011} over a finite field of characteristic $2$ can be straightforwardly generalized to the case of finite fields of odd characteristic as follows.
 
 \begin{lemma} \label{LemQx} Let $\nu\geq 0$  be an integer and let $d$ be an integer given  in \eqref{eq-firstfact}.  Then following statements holds.
     \begin{enumerate}[$1)$]
         \item $  d\in G_{(p^l,1)}(\nu+1)$ if and only if  every irreducible factor of $Q_{d2^{\nu+1}}(x)$ is self-reciprocal.
         \item $  d\notin G_{(p^l,1)}(\nu+1)$ if and only if the irreducible factors of $Q_{d2^{\nu+1}}(x)$ form reciprocal polynomial pairs.
     \end{enumerate}
     
 \end{lemma}
 
 Applying Lemma \ref{LemQx}  and Equation  \eqref{eq-firstfact}  (cf. \cite[Equation (29)]{SJLU2015}), it can be concluded that

 \begin{align}\label{all}
 x^n+1&=\left( \prod _{\substack{d|n^\prime   \\ d\in G_{(p^l,1)}(\nu+1)}} Q_{d2^{\nu+1}}(x)  \prod _{\substack{d|n^\prime   \\ d\notin G_{(p^l,1)}(\nu+1)}}  Q_{d2^{\nu+1}}(x)  \right)^{p^r}  \notag\\
 &=\left( \prod_{\substack{d|n^\prime   \\ d\in G_{(p^l,1)}(\nu+1)} }\prod_{i=1}^{\rho(d2^{\nu+1} ,p^l)}f_{di}(x)\prod_{\substack{d|n^\prime   \\ d\notin G_{(p^l,1)}(\nu+1)}  }\prod_{j=1}^{\sigma(d2^{\nu+1},p^l)}g_{dj}(x)g^*_{dj}(x)\right)^{p^r},
 \end{align}
 where $f_{di}(x)$ is  a  self-reciprocal irreducible  polynomial for all $d$ and $i$,    $g_{dj}(x)$ and  $g_{dj}^*(x)$  are a reciprocal irreducible polynomial pair   for all  $d$ and $j$, $
 \rho(d2^{\nu+1},p^l)=
 \frac{\phi(d2^{\nu+1})}{{\rm  ord}_{d2^{\nu+1}}(p^l)}$ and $
 \sigma(d2^{\nu+1},p^l)=
 \frac{\phi(d2^{\nu+1})}{2{\rm  ord}_{d2^{\nu+1}}(p^l)}.$
 
 The existence of a Euclidean self-dual  negacyclic code of length $n$ over $\mathbb{F}_{p^l}$ can be determined using \eqref{all} as follows.
 
 \begin{proposition}\label{prop:charE}
     Let $p$ be  an odd prime and let $n=2^\nu p^rn^\prime$, where $\nu\geq 0$ and $r\geq 0$ are integers and $ n^\prime$ is an odd positive integer such that $p\nmid n^\prime$. Let $l$ be a positive integer. Then there exists  a  Euclidean self-dual  negacyclic code of length $n$ over $\mathbb{F}_{p^l} $ if and only if  $\nu>0$ and
     $d\notin G_{(p^l,1)}(\nu+1)$   for all $d|n^\prime$.
 \end{proposition}
 \begin{proof}
     Assume  that there exists  a  Euclidean self-dual  negacyclic code  $C$ of length $n$ over $\mathbb{F}_{p^l} $. Then $n=2^\nu p^rn^\prime$  must be even which implies that $\nu>0$.  Let  $g(x)$ be the generator polynomial for $C$  and let $h(x):=\frac{x^n+1}{g(x)}$.      Since  $C$ is Euclidean self-dual, we have
     $g(x)=h^*(x)$.
     Suppose that there exists  a positive integer $d$ such that $d|n^\prime$ and  $d\in G_{(p^l,1)}(\nu+1)$.        Then $f_{d1}(x)$ has the same multiplicity $m $ in $g(x)$ and in $g^*(x)=h(x)$. It follows that  the  multiplicity of $f_{d1}(x)$ in $x^n+1$ is $2m =p^r$, a contradiction.  Hence, $d\notin G_{(p^l,1)}(\nu+1)$ for all divisors $d$ of  $n^\prime$.

     Conversely, assume that  $\nu>0$ and
     $d\notin G_{(p^l,1)}(\nu+1)$   for all $d|n^\prime$.  From \eqref{all},     we have
     \[ x^n+1=\prod_{\substack{d|n^\prime  }  }\prod_{j=1}^{\sigma(d2^{\nu+1},p^l)}\left(  g_{dj}(x)g^*_{dj}(x)\right)^{p^r}.\]
     It is not difficult to see that the negacyclic code generated by \[g(x)=\prod_{\substack{d|n^\prime  }  }\prod_{j=1}^{\sigma(d2^{\nu+1},p^l)} \left(  g_{dj}(x)\right)^{p^r}\] is Euclidean self-dual.
 \end{proof}

 % From Proposition~\ref{prop:charE}, to study  Euclidean self-dual negacyclic codes,  it is sufficient to focus on the case where $n$ is even, or equivalently, $v>0$. The conditions in Proposition~\ref{prop:charE} can be simplified  as follows. The result coincides with  \cite[Theorem 3]{B2008}.
 
 The result of Blackford \cite[Theorem 3]{B2008} can be obtained as a corollary to Proposition~\ref{prop:charE} as follows

 \begin{corollary}\label{charE}
     Let $p$ be an odd prime and let $n=2^\nu p^rn^\prime$, where $\nu >0$ and $r\geq 0$ are integers and $ n^\prime$ is an odd positive integer such that $p\nmid n^\prime$. Let $l$ be a positive integer. Then there exists  a  Euclidean self-dual negacyclic code of length $n$ over $\mathbb{F}_{p^l} $ if and only if $1 \notin   G_{(p^l,1)}(\nu+1)$ i.e. if and only if $2^{\nu+1}\nmid (p^l+1) $.
 \end{corollary}
 
 \begin{proof}
     Suppose   there exists a   Euclidean self-dual negacyclic code of length $n$ over $\mathbb{F}_{p^l} $. Then, by Proposition~\ref{prop:charE}, we have   that  $ d\notin G_{(p^l,1)}(\nu+1)$   for all $d|n^\prime$, in particular   $ 1 \notin G_{(p^l,1)}(\nu+1)$. Conversely if $ 1 \notin G_{(p^l,1)}(\nu+1)$, then $d\notin G_{(p^l,1)}(\nu+1)$ for  all odd numbers $d$, in particular for all divisors $d$ of $n^\prime$,
     by Corollary~\ref{corBad}. The result therefore follows from Proposition~\ref{prop:charE}.
     
     From Corollary~\ref{corBad}, the last two statements  are equivalent.
 \end{proof}

 \begin{corollary}
     Let $p$ be an odd prime and let   $\nu >0$ and $r\geq 0$ be  integers.  Let $ n^\prime$ be  an odd positive integer such that $p\nmid n^\prime$. Let $l$ be a positive integer.  If there exists  a  Euclidean self-dual negacyclic code of length $2^\nu p^rn^\prime$ over $\mathbb{F}_{p^l} $, then there exist
     a  Euclidean self-dual negacyclic code of length $2^\mu p^rn^\prime$ over $\mathbb{F}_{p^l} $   for all positive integers   $\mu\geq \nu$.
 \end{corollary}

 Using the above results, we can prove two corollaries which are due to  \cite[Theorem 3.3]{LI2016}.
 \begin{corollary}
     Let $p$ be an odd prime and $l$ be a positive integer. Assume that  $2^2| (p^l-1)$. Then there exists a Euclidean self-dual negacyclic code of length $2^\nu p^r$ over $\mathbb{F}_{p^l}$ if and only if $\nu \geq 1$.
 \end{corollary}
 \begin{proof} As $2^2| (p^l-1)$, let $p^\ell-1= 2^\gamma c$ for some odd integer $c$ and integer $\gamma\geq 2$. Then $p^\ell+1= 2^\gamma c+2=2(2^{\gamma-1}c+1)$, where $2^{\gamma-1}c+1$ is odd. We see that $2^{\nu+1} \nmid (p^l+1)$   if and only if $\nu \geq 1$. Then, the result follows from Corollary~\ref{charE}.
     % By Corollary~\ref{charE}, {there exists} a Euclidean self-dual negacyclic code  of length $2^\nu p^r$ over $\mathbb{F}_{p^l}$ if and only if  $2^{\nu+1} \nmid (p^l+1)$. Now $2^{\nu+1} \nmid (p^l-1)$   if and only if $\nu \geq 1$.
     %     Since $2^2| (p^l-1)$, we have $p^\ell+1= 2^\gamma c+2=2(2^{\gamma-1}c+1)$ for some odd integer $c$ and integer $\gamma\geq 2$. Hence,   $2^{\nu+1} \nmid (p^l+1)$   if and only if $\nu \geq 1$. Hence, The result follows.
 \end{proof}

 \begin{corollary}
     Let $p$ be an odd prime and $l$ be a positive integer. Let $\gamma$ be a positive integer such that $\gamma \geq 2$ and $2^\gamma || (p^l+1)$. Then there exists a Euclidean self-dual negacyclic code of length $2^\nu p^r$ over $\mathbb{F}_{p^l}$ if and only if $\nu \geq \gamma$.
 \end{corollary}
 
 \begin{proof}
     
     Since $2^\gamma || (p^l+1)$, we see that  $2^{\nu+1} \nmid (p^l+1)$ if and only if   $\nu  \geq \gamma$. The desired result  therefore follows by  Corollary~\ref{charE}.
 \end{proof}

 The explicit number of Euclidean self-dual negacyclic codes of specific lengths $n=2^{\nu}$, $\nu \geq 1$ and  $n=2q^t$ ($ \nu =1$, $n^\prime=q^t, q$ an odd prime) was obtained   in  \cite[Theorems 3,4]{BR2012}, \cite[Theorems 3,4]{BR2013} respectively. The idea for a general formula for  the number of Euclidean self-dual negacyclic codes of length $n$ over $\mathbb{F}_{p^{l} }$  was given in \cite[Corollary 2.6]{LI2016}. However, there are no explicit  formulas for general $n$. Using $2^\beta$-good integers discussed in Section 2,  a general formula for the number of Euclidean self-dual negacyclic codes of length $n$ over $\mathbb{F}_{p^{l} }$  can be deduced.
 
 \begin{theorem} \label{thmNumE1}
     Let $p$ be  an odd prime and let $n=2^\nu p^rn^\prime$, where $\nu>0$ and $r\geq 0$ are integers and $ n^\prime$ is an odd positive integer such that $p\nmid n^\prime$. Let $l$ be a positive integer.  The number of Euclidean self-dual negacyclic codes of length $n$ over $\mathbb{F}_{p^{l} }$ is
     \begin{align}
     NE(p^l,n) :=
     \begin{cases}
     (p^r+1)^{\frac{1}{2}\sum_{d| n^\prime} \frac{\phi(d2^{\nu+1})}{ {\rm ord}_{d2^{\nu+1}}(p^{l})}}& \text{  if  }2^{\nu+1} \nmid (p^l+1),\\
     0& \text{ otherwise}.
     \end{cases}\label{ESDcc}
     \end{align}
 \end{theorem}
 \begin{proof} If  $2^{\nu+1} | (p^l+1)$, there are no Euclidean self-dual negacyclic codes of length $n$ over $\mathbb{F}_{p^{l} }$ by Corollary~\ref{charE}. Hence, $  NE(p^l,n)=0$.
     
     Assume that  $2^{\nu+1} \nmid (p^l+1)$.
     From \eqref{all},     we have
     \[ x^n+1=  \prod_{\substack{d|n^\prime  }  }\prod_{j=1}^{\sigma(d2^{\nu+1},p^l)}\left(g_{dj}(x)g^*_{dj}(x)\right)^{p^r}.\]
     Let \[g(x)=\prod_{\substack{d|n^\prime  }  }\prod_{j=1}^{\sigma(d2^{\nu+1},p^l)}g_{dj}(x)^{a_{dj}}g^*_{dj}(x)^{b_{dj}}\] be  the generator polynomial of a Euclidean self-dual cyclic code of length $n$ over $\mathbb{F}_{p^{l} }$, where $0\leq a_{dj}\leq p^r$ and $0\leq b_{dj}\leq p^r$.  Then
     
     \[g(x)=h^*(x)=\prod_{\substack{d|n^\prime  }  }\prod_{j=1}^{\sigma(d2^{\nu+1},p^l)}g_{dj}(x)^{p^r-b_{dj}}g^*_{dj}(x)^{p^r-a_{dj}}.\]
     This implies that   $a_{dj}+b_{dj}=p^r$, and hence, the number of  choices for $(a_{dj},b_{dj})$ is $p^r+1$ for all  $d\mid n^\prime$ and $1\leq j\leq \sigma(d2^{\nu+1},p^l)$.   Therefore, the formula is proved.
 \end{proof}
 
 From the above theorem, the remaining difficult part is to compute
 \begin{align}\label{eq:t} t(n'2^{\nu+1},p^l) := {\frac{1}{2}\sum_{d| n^\prime}\frac{\phi(d2^{\nu+1})}{ {\rm ord}_{d2^{\nu+1}}(p^{l})}}\end{align}
 which  is independent of a factor $p^r$ of $n$.  Some results on a specific $n^\prime$ are given in 3.2.2.

 \subsubsection{Euclidean Self-Dual Negacyclic Codes of Lengths $2^\nu$ and  $2^\nu p^r$}
 
 In this part, we give  explicit formulae of Euclidean self-dual  negacyclic codes of lengths $2^\nu$ and  $2^\nu p^r$ over $\mathbb{F}_{p^l}$, where $\nu$ is a positive integer. First, we compute $\operatorname{ord}_{2^{\nu+1}}(p^l)$ which is a key to determine  $t(2^{\nu+1},p^l) $.
 
 %From necessary conditions for the existence of  a Euclidean self-dual  negacyclic code over $\mathbb{F}_{p^l}$ in Proposition~\ref{prop:charE}, we assume that $\nu>0$ is an integer. First, we compute $\operatorname{ord}_{2^{\nu+1}}(p^l)$ which is a key to determine  $t(2^{\nu+1},p^l) $. Subsequently, the  formulas for the number of Euclidean self-dual  negacyclic codes of length $2^\nu$ and  $2^\nu p^r$ are deduced.

 First, useful number theoretical results   are recalled.
 
 \begin{theorem}[{\cite[Theorem 3.10]{N2000}}] \label{ordermod2}
     If $k\geq 3$, then $5$ has order $2^{k-2} $ modulo ${2^k}$.
     If $a \equiv 1\mod{4}$ then there exists a unique integer $i \in \{0,1,\dots, 2^{k-2}-1\}$ such that
     $a \equiv 5^i \mod{2^k}$.
     If $a \equiv 3\mod{4}$ then there exists a unique integer $i \in \{0,1,\dots, 2^{k-2}-1\}$ such that
     $a \equiv -5^i \mod{2^k}$.
 \end{theorem}
 
 For convenience, for each  $\nu\geq 1$, let $\alpha_p$ denote the unique integer in the set $\{0,1,\dots, 2^{\nu-1}-1\}$ such that
 \begin{align} \label{alpha}\begin{cases}
 p \equiv 5^{\alpha_p}\mod{2^{\nu+1}} &\text{ if } p \equiv 1\mod{4},\\
 p \equiv -5^{\alpha_p}\mod{2^{\nu+1}} &\text{ if } p \equiv 3\mod{4}.
 \end{cases}
 \end{align}
 Note that $\alpha_p=0$ if $\nu=1$ and the existence of $\alpha_p$ is guaranteed by  Theorem \ref{ordermod2} for all  $\nu\geq 2$.
 
 \begin{lemma}\label{formula1}
     Let $\nu$ and $ l$ be positive integers and let  $p$ be an odd prime.   Then
     $$\operatorname{ord}_{2^{\nu+1}}(p^l) =\begin{cases}
     1 & { \text{ if } \nu=1 \;\text{and} \;p^l \equiv 1 \mod{4},}\\
     2 & { \text{ if } \nu=1 \;\text{and} \;p^l \equiv 3 \mod{4},}\\
     \frac{2^{\nu-1}}{\operatorname{gcd}(2^{\nu-1},\alpha_pl)} &\text{ if } \nu\geq 2.
     \end{cases} $$
 \end{lemma}

 \begin{proof}
     It is obvious  for the case $\nu = 1$. Let now $\nu \geq 2$. By Theorem~\ref{ordermod2}, there exists $\alpha_p \in \{0,1,2,\dots, 2^{\nu-1}-1\}$ such that $p \equiv \pm 5^{\alpha_p} \mod{2^{\nu+1}}$.
     If $p \equiv 1\mod{4}$ or $p \equiv 3\mod{4}$ and $l$ is even,  then
     $\operatorname{ord}_{2^{\nu+1}}(p^l) = \operatorname{ord}_{2^{\nu+1}}(5^{\alpha_pl})$.

     When $p \equiv 3 \mod{4}$ and $l$ is odd, we have $p^l \equiv - 5^{\alpha_pl} \mod{2^{\nu+1}}$ and 
     
     \[ \operatorname{ord}_{2^{\nu+1}}(5^{\alpha_pl})= \frac{2^{\nu-1}}{\operatorname{gcd}(2^{\nu-1},\alpha_pl)}\]
     
     \noindent  which  is even as $\alpha_p < 2^{\nu-1}$.
     Thus $\operatorname{ord}_{2^{\nu+1}}(p^l) = \operatorname{lcm}(\operatorname{ord}_{2^{\nu+1}}(-1), \operatorname{ord}_{2^{\nu+1}}(5^{\alpha_pl})) = \operatorname{ord}_{2^{\nu+1}}(5^{\alpha_pl}).$
     In all cases, we therefore  have
     $$\operatorname{ord}_{2^{\nu+1}}(p^l) = \operatorname{ord}_{2^{\nu+1}}(5^{\alpha_pl})= \frac{2^{\nu-1}}{\operatorname{gcd}(2^{\nu-1},\alpha_pl)}$$
     as desired.
 \end{proof}
 
 \begin{corollary}\label{corE1}
     Let $p$ be an odd prime   and let $\nu$ and $l$ be positive integers. {If $2^{\nu+1} \nmid (p^l+1)$ then }
     $$t(2^{\nu+1},p^l) =
     {\operatorname{gcd}(2^{\nu-1},\alpha_pl)}
     $$
     In particular,  if $p=5$,  then
     $$ t(2^{\nu+1},5^l) = {\operatorname{gcd}(2^{\nu-1},l)}.$$
     
 \end{corollary}

 \begin{proof}
     Since $t(2^{\nu+1},p^l) = \frac{\phi(2^{\nu+1})}{2\operatorname{ord}_{2^{\nu+1}}{(p^l)}}$, it follows that   $t(2^{\nu+1},p^l) = {\operatorname{gcd}(2^{\nu-1},\alpha_pl)}$ by  Lemma \ref{formula1}.
     If $p=5$,  then $\alpha_5=1$.  Hence,  $t(2^{\nu+1},5^l) = {\operatorname{gcd}(2^{\nu-1},l)}$ as desired.
 \end{proof}

 Combining Theorem~\ref{thmNumE1} and Corollary~\ref{corE1},  the  numbers of Euclidean self-dual  negacyclic codes of lengths $2^\nu$ and  $2^\nu p^r$ over $\mathbb{F}_{p^l}$ can be given as follows.
 \begin{corollary}
     Let $\nu$ and $l$ be positive integers  and let $p$ be  an odd prime.  Let $r\geq 0$ be an   integer. { If $2^{\nu+1} \nmid (p^l+1)$,  then }
     \begin{align*}
     NE(p^l,2^\nu p^r) =
     (p^r+1) ^{\operatorname{gcd}(2^{\nu-1},\alpha_pl)} %&
     \end{align*}
     Otherwise, $  NE(p^l,2^\nu p^r) =0$.
 \end{corollary}

 \subsection{Hermitian Self-Dual Negacyclic Codes}
 
 In this subsection, we focus on the characterization and enumeration of Hermitian self-dual negacyclic codes over $\mathbb{F}_{p^{2l}}$. General properties of Hermitian self-dual negacyclic codes of any lengths are given in 3.3.1.  Hermitian self-dual negacyclic codes  of some specific lengths are discussed in  3.3.2.

 \subsubsection{Hermitian Self-Dual Negacyclic Codes}

 The result in  \cite[Lemma 3.5]{JLS2013} over a finite field of characteristic $2$ can be straightforwardly generalized to the case of  polynomials over a finite field of odd characteristic as follows.
 
 \begin{lemma} \label{LemQx2} Let $\nu\geq 0$  be an integer and let $d$ be an integer given in \eqref{eq-firstfact}. Then the  following statements holds.
     \begin{enumerate}[$1)$]
         \item $d\in OG_{(p^l,1)}(\nu+1)$ if and only if  every irreducible factor  of $Q_{d2^{\nu+1}}(x)$ is self-conjugate-reciprocal.
         \item $d\notin OG_{(p^l,1)}(\nu+1)$ if and only if the irreducible factors of $Q_{d2^{\nu+1}}(x)$ form conjugate-reciprocal polynomial pairs.
     \end{enumerate}
     
 \end{lemma}
 
 Applying Lemma \ref{LemQx2} and Equation  \eqref{eq-firstfact}  (cf. \cite[Equation (29)]{SJLU2015}), it can be concluded that
 
 \begin{align}\label{all2}
 x^n+1&=\left( \prod _{\substack{d|n^\prime   \\ d\in OG_{(p^l,1)}(\nu+1)}} Q_{d2^{\nu+1}}(x)  \prod _{\substack{d|n^\prime   \\ d\notin OG_{(p^l,1)}(\nu+1)}}  Q_{d2^{\nu+1}}(x)  \right)^{p^r}  \notag\\
 &=\left( \prod_{\substack{d|n^\prime   \\ d\in OG_{(p^l,1)}(\nu+1)} }\prod_{i=1}^{\mu(d2^{\nu+1} ,p^{2l})}f_{di}(x)\prod_{\substack{d|n^\prime   \\ d\notin OG_{(p^l,1)}(\nu+1)}  }\prod_{j=1}^{\eta(d2^{\nu+1},p^{2l})}g_{dj}(x)g^\dagger _{dj}(x)\right)^{p^r},
 \end{align}
 where $f_{di}(x)$ is  a  self-conjugate-reciprocal irreducible  polynomial for all $d$ and $i$,    $g_{dj}(x)$ and  $g_{dj}^\dagger (x)$  are a conjugate-reciprocal irreducible polynomial pair   for all  $d$ and $j$, $
 \mu(d2^{\nu+1},p^{2l})=
 \frac{\phi(d2^{\nu+1})}{{\rm  ord}_{d2^{\nu+1}}(p^{2l})}$ and $
 \eta(d2^{\nu+1},p^{2l})=
 \frac{\phi(d2^{\nu+1})}{2{\rm  ord}_{d2^{\nu+1}}(p^{2l})}.$

 The existence of  a Hermitian self-dual  negacyclic code of length $n$ over $\mathbb{F}_{p^{2l}}$ can be determined using \eqref{all2} as follows.

 \begin{proposition} \label{prop:charH}
     Let $p$ be  an odd prime and let $n=2^\nu p^rn^\prime$, where $\nu\geq 0$ and $r\geq 0$ are integers and $ n^\prime$ is an odd positive integer such that $p\nmid n^\prime$. Let $l$ be a positive integer. There there exists  a  Hermitian self-dual  negacyclic code of length $n$ over $\mathbb{F}_{p^{2l}} $ if and only if  $\nu>0$ and
     $d\notin OG_{(p^l,1)}(\nu+1)$   for all $d|n^\prime$.
 \end{proposition}
 
 \begin{proof}
     Assume  that there exists  a  Hermitian self-dual  negacyclic code  $C$ of length $n$ over $\mathbb{F}_{p^{2l}} $. Then $n=2^\nu p^rn^\prime$  must be even which implies that $\nu>0$.  Let  $g(x)$ be the generator polynomial for $C$  and let $h(x):=\frac{x^n+1}{g(x)}$.      Since  $C$ is Hermitian self-dual, we have
     $g(x)=h^\dagger(x)$.
     Suppose that there exists  a positive integer $d$ such that $d|n^\prime$ and $d\in OG_{(p^l,1)}(\nu+1)$.        Then $f_{d1}(x)$ has the same multiplicity $m $ in $g(x)$ and in $g^\dagger(x)=h(x)$. It follows that  the  multiplicity of $f_{d1}(x)$ in $x^n+1$ is $2m=p^r$, a contradiction.

     Conversely, assume that  $\nu>0$ and
     $d\notin OG_{(p^l,1)}(\nu+1)$   for all $d|n^\prime$.  From \eqref{all2},     we have
     \[ x^n+1= \prod_{\substack{d|n^\prime  }  }\prod_{j=1}^{\eta(d2^{\nu+1},p^{2l})}\left( g_{dj}(x)g^\dagger_{dj}(x)\right)^{p^r}.\]
     It is not difficult to see that the negacyclic code of length $n$ generated by \[g(x)=  \prod_{\substack{d|n^\prime  }  }\prod_{j=1}^{\eta(d2^{\nu+1},p^{2l})}\left( g_{dj}(x)\right)^{p^r}\] is Hermitian self-dual.
 \end{proof}

 For $\nu>0$ and an odd positive  integer $d$,   we have     $d\notin OG_{(p^l,1)}(\nu+1)$  if and only if     $d\notin G_{(p^l,1)}(\nu+1)$   by Corollary~\ref{odd-good2}.  Then the conditions in Proposition~\ref{prop:charH} can be simplified using the above discussion and Corollary~\ref{charE} as follows.
 \begin{corollary} Let $p$ be  an odd prime and let $n=2^\nu p^rn^\prime$, where $\nu>0$ and $r\geq 0$ are integers and $ n^\prime$ is an odd positive integer such that $p\nmid n^\prime$.  Let $l$ be a positive integer. There there exists  a  Hermitian self-dual  negacyclic code of length $n$ over $\mathbb{F}_{p^{2l}} $ if and only if $2^{\nu+1}\nmid (p^l+1) $.
 \end{corollary}

 A general formula for the number of Hermitian self-dual negacyclic codes of length $n$ over $\mathbb{F}_{p^{2l}}  $ can be summarized as follows.
 \begin{theorem}\label{thmNumH1}
     Let $p$ be  an odd prime and let $n=2^\nu p^rn^\prime$, where $\nu>0$ and $r\geq 0$ are integers and $ n^\prime$ is an odd positive integer such that $p\nmid n^\prime$. Let $l$ be a positive integer.  The  number of Hermitian self-dual negacyclic codes of length $n$ over $\mathbb{F}_{p^{2l}} $ is
     \begin{align}
     NH(p^{2l},n):=\begin{cases}  (p^r+1)^{\frac{1}{2}\sum_{d| n^\prime}\frac{\phi(d2^{\nu+1})}{ {\rm ord}_{d2^{\nu+1}}(p^{2l})}} & \text{ if }  2^{\nu+1}\nmid (p^l+1),\\
     0 & \text{ otherwise}.
     \end{cases}\label{HSDcc}
     \end{align}
 \end{theorem}
 \begin{proof}
     By replacing  $G_{(p^l,1)}(\nu+1)$ with $ OG_{(p^l,1)}(\nu+1)$  in  the proof of Theorem~\ref{thmNumE1}, the result can be deduced.
 \end{proof}
 
 Next, we focus on  the determination of
 \begin{align}\label{eq:tau} \tau (n'2^{\nu+1},p^l) :=  {\frac{1}{2}\sum_{d| n^\prime}\frac{\phi(d2^{\nu+1})}{ {\rm ord}_{d2^{\nu+1}}(p^{2l})}},
 \end{align}
 for some specific $n^\prime$. Note that  $\tau (n'2^{\nu+1},p^l)$
 is  independent of the  factor $p^r$ of $n$.

 \subsubsection{Hermitian Self-Dual Negacyclic Codes of Lengths $2^\nu$ and  $2^\nu p^r$}

 From  Proposition~\ref{prop:charH}, we assume that $\nu>0$ is an integer and conclude the following results.
 \begin{proposition}\label{corH1}
     Let $p$ be an odd prime   and let $\nu$ and $l$ be positive integers.
     If $2^{\nu+1}\nmid (p^l+1)$ then
     $$\tau(2^{\nu+1},p^l)=
     {\operatorname{gcd}(2^{\nu-1},2\alpha_pl)}
     $$
     In particular,  if $p=5$,  then
     $$ \tau (2^{\nu+1},5^l) = {\operatorname{gcd}(2^{\nu-1},2l)}.$$
     
 \end{proposition}

 \begin{proof}
     The result follows from the fact that  $\operatorname{ord}_{2^{\nu}+1}(p^{2l}) = \operatorname{gcd}({2^{\nu-1},2\alpha_pl})$.
 \end{proof}
 
 Combining Theorem~\ref{thmNumH1} and  Proposition~\ref{corH1},  the  number of Hermitian self-dual  negacyclic codes of lengths $2^\nu$ and  $2^\nu p^r$ over $\mathbb{F}_{p^{2l}}$  is determined as follow.
 \begin{corollary}
     Let $\nu$ and $l$ be positive integers  and let $p$ be  an odd prime.  Let $r\geq 0$ be an integer.
     If $2^{\nu+1}\nmid (p^l+1)$,  then
     \begin{align*}
     NH(p^{2l},2^\nu p^r) =
     (p^r+1) ^{\operatorname{gcd}(2^{\nu-1},2\alpha_pl)}
     \end{align*}
     Otherwise, $ NH(p^{2l},2^\nu p^r) =0$.

 \end{corollary}

\end{document}